\documentclass{article}
\usepackage{algorithm}

\usepackage{}
\usepackage{enumerate}
\usepackage{mathbbold}
\usepackage{bbm}
\usepackage{amsfonts}
\usepackage[T1]{fontenc}
\usepackage{latexsym,amssymb,amsmath,amsfonts,amsthm}
\usepackage{graphicx, color}
\usepackage{subfigure}
\usepackage{epsf}
\usepackage{epstopdf}
\usepackage{amssymb}
\usepackage{authblk} 
\usepackage{amsthm}
\usepackage{multirow}
\usepackage{diagbox}
\usepackage{stmaryrd}
\usepackage{mathrsfs}
\usepackage{hyperref}
\hypersetup{
colorlinks =true, 
linkcolor=black, 
anchorcolor=black,
citecolor=black}

\numberwithin{equation}{section}

\newcommand{\R}{{\mathbb R}}

\newcommand{\be}{\begin{eqnarray}}
\newcommand{\ben}{\begin{eqnarray*}}
\newcommand{\en}{\end{eqnarray}}
\newcommand{\enn}{\end{eqnarray*}}

\newcommand{\pa}{\partial}

\newcommand{\Om}{\Omega}

\newtheorem{theorem}{Theorem}[section]

\begin{document}
\fontsize{9}{11}\selectfont
\title{\bf A Radon-transform-based formula for reconstructing acoustic sources from the scattered fields}

\author[1]{Xiaodong Liu}
\author[1,2]{Jing Wang} 
\makeatletter 
\renewcommand{\AB@affilnote}[1]{\footnotemark[#1]} 
\renewcommand{\AB@affillist}[1]{}
\renewcommand{\AB@affillist}[2]{} 
\makeatother

\date{}
\maketitle
\footnotetext[1]{State Key Laboratory of Mathematical Sciences, Academy of Mathematics and Systems Science,
Chinese Academy of Sciences, Beijing 100190, China. Email: xdliu@amt.ac.cn} 
\footnotetext[2]{Corresponding author. School of Mathematical Sciences, University of Chinese Academy of Sciences, Beijing 100049, China, Email:wangjing23@amss.ac.cn}

\begin{abstract}
We propose a novel indicator function for reconstructing acoustic sources from multi-frequency near-field measurements. The theoretical basis is established by a formula relating the scattered field to the source function through the Radon transform. Such a representation enables us to recover the source function directly. The efficiency and robustness of the novel indicator function are verified by several numerical examples.

\vspace{.2in}

\end{abstract}

\section{Introduction}

In recent years, numerous research methods for reconstructing acoustic sources have been proposed.
These methods can be broadly classified into two categories: qualitative methods and quantitative methods. The extended source itself cannot be uniquely determined from sparse measurements \cite{R3source-Meng-Liu}, and only partial information about the source function is available, including the size and location of its support. 
Accordingly, all existing qualitative methods are established on the basis of the uniqueness results derived under sparse observation data. The studies in \cite{direct-sampling-Liu-Hu, R3source-Meng-Liu} propose algorithms for reconstructing the size and location of the source support using the direct sampling method, while the work presented in \cite{factorization method} develops the factorization method. A recent study in \cite{extended source-sparse-Liu-Shi} achieves the accurate reconstruction of the shape and location of annular sources and polygonal sources from sparse measurement data. 
In contrast, when full-aperture measurement data are accessible, it becomes possible in principle to retrieve complete information about the source function. Existing quantitative methods for this purpose can be further divided into iterative and non‑iterative methods.
The research presented in \cite{continuation method-Bao, recursive-algorithm-Bao} introduces two iterative approaches, namely the continuation method and the recursive method.
Representative non-iterative methods include the Fourier method presented in \cite{Fourier-Wang-Guo-far,Fourier-zhang-Guo-near} and the direct sampling method proposed in \cite{extended source-sparse-Liu-Shi}.

However, a simple, direct equality linking the multi-frequency near-field scattering data to the source function itself remains elusive, particularly one that simultaneously provides both geometrical support and quantitative amplitude recovery without iterative forward solvers. Bridging this gap is the primary objective of this paper.
Compared to the iterative methods \cite{recursive-algorithm-Bao,continuation method-Bao}, the proposed approach bypasses the need for solving the direct problem, thereby offering a more direct and computationally efficient implementation.
In contrast to the Fourier method designed for the near-field regime presented in \cite{Fourier-zhang-Guo-near}, our method avoids computing normal derivative data and does not impose special constraints on the choice of frequencies.
Furthermore, unlike the far-field case studied in \cite{extended source-sparse-Liu-Shi}, the relationship between near-field data and the source function does not reduce to a simple inverse Fourier transform, thereby rendering theoretical analysis in the near-field regime considerably more complex.

The problem model investigated in this paper is as follows. In an isotropic homogeneous medium, the time-harmonic wave 
$u^s$ radiated by $S$ satisfies the Helmholtz equation
\begin{equation*}
\Delta u^s +k^2u^s =-S \quad \text{in}\,\, \R^2, 
\end{equation*}
where $k>0$ is the wave number, $S \in L^2(\R^2)$ is a real-valued function with compact support $\Omega\in B_R(0)$.
Throughout this paper, we denote by $B_R(x)$ the disk centered at $x$ with radius $R$, and by $\pa B_R(x)$ its boundary.
Moreover, the scattered field $u^s$ satisfies the Sommerfeld radiation condition
\begin{equation*}
\lim\limits_{r=|x|\rightarrow \infty}r^{\frac{1}{2}}\left(\frac{\pa u^s}{\pa r}-iku^s\right)=0, \quad x\in \R^2,
\end{equation*}
uniformly for all $x/|x|$.
In this paper, we focus on the following inverse source problem, denoted as \textbf{(IP-near-field)}
\begin{itemize}
    \item Establish a direct and simple equality relation between $u^s$ and the source function $S$.
\end{itemize}

First, we establish an explicit relationship between the scattered field and the Radon transform of the source function. Building on the results in \cite{radon transform-07}, we then derive a direct equality between the source function and the scattered field. This implies that, given sufficient measurement data, we can not only reconstruct the location and shape of the support domain $\Omega$, but also recover the corresponding function values within it.
 
The rest of the paper is organized as follows. In Section \ref{Numerical methods}, we construct the indicator function for \textbf{(IP-near-field)} and develop a corresponding algorithm based on the direct sampling method. In Section \ref{Numerical examples}, the efficiency and robustness of the proposed algorithm are demonstrated through three numerical examples of distinct types.


\section{A novel indicator}
\label{Numerical methods}
\setcounter{equation}{0}
Let $J_n, Y_n$ and $H_n^{(1)}$ denote, respectively, Bessel, Neumann and Hankel functions of the first kind of order $n$.
The scattered field is then given by
\begin{align}\label{Scattered field}
u^s(x,k)=\int_{\R^2}\Phi_k(x,y)S(y)dy,\quad x\in\mathbb{R}^2,
\end{align}
where 
\begin{equation*}
\Phi_k(x,y):=\frac{i}{4}H_0^{(1)}(k|x-y|), x,y \in \R^2, x\neq y.
\end{equation*}
We then define the following indicator function 
\begin{equation}
\begin{aligned}
I_S(z):=&\frac{R}{2\pi}\int_0^{2\pi} \Bigg[  \frac{z-x}{|z-x|}\cdot (\cos\theta, \sin\theta)\Bigg]\int _{0}^{+\infty}k ^2\Big[\Im (u^s(x,k))N_1(k |z-x|)\\
&+\Re (u^s(x,k))J_1(k |z-x|)\Big]dk d\theta, \quad z\in B_R(0),
\label{acoustic-source-integral}
\end{aligned}
\end{equation}
where $x:=R(\cos (\theta), \sin (\theta))$. 
The following theorem shows that the indicator function $I_{S}(z)$ defined in  \eqref{acoustic-source-integral} is exactly the source function $S(z)$, which then gives a formula to compute the source function.

\begin{theorem}
\label{source-uniqueness}
Let $S(z)\in L^2(\R^2)$ be a real-valued source function with compact support $\Omega\subset B_R(0)$, then
\begin{equation*}
I_S(z)=S(z),\quad z\in B_R(0).
\end{equation*}
\end{theorem}
\begin{proof}
Recall the following two identities \cite{radon transform-07}
\begin{equation}
S(z)=\frac{1}{2\pi}\int_0^{+\infty}\int _{\R^2}S(y)J_0(k|z-y|)dy k dk,\quad z\in  B_R(0),
\label{07-(1)}
\end{equation}
and
\begin{equation}
J_0(k|z-y|)=\frac{1}{4}\int_{\pa B_R(0)}\Bigg[J_0(k |x-y|)\frac{\pa N_0(k |z-x|)}{\nu(x)}-N_0(k |x-y|)\frac{\pa J_0(k |z-x|)}{\nu(x)}\Bigg]ds(x).
\label{07-(5)}
\end{equation}
Substituting equation \eqref{07-(5)} into \eqref{07-(1)} yields  the following expression for $S(z)$
\begin{equation}
\begin{aligned}
S(z)=&\frac{1}{8\pi}\int_{0}^{+\infty}k \int_{\pa B_R(0)}\Bigg[\frac{\pa N_0(k |z-x|)}{\nu(x)}\left(\int_{B_R(0)}J_0(k |x-y|)S(y)dy\right)\\
&-\frac{\pa J_0(k |z-x|)}{\nu(x)}\left(\int_{B_R(0)}N_0(k |x-y|)S(y)dy\right)\Bigg] ds(x) dk ,\quad z\in B_R(0),
\label{07-S(z)}
\end{aligned}
\end{equation}
Inserting \eqref{Scattered field} into \eqref{acoustic-source-integral}, then using equation \eqref{07-S(z)} and the following differentiation equations
\begin{equation*}
\nabla J_0(|z|)=-\frac{z}{|z|}J_1(|z|),\quad \nabla N_0(|z|)=-\frac{z}{|z|}N_1(|z|),\quad z\in \R^2,
\end{equation*}
we derive
\begin{align}
I_S(z)=&\frac{R}{8\pi}\int_0^{2\pi} \Bigg[  \frac{z-x}{|z-x|}\cdot (\cos\theta, \sin\theta)\Bigg]\int _{0}^{+\infty}k ^2\Bigg[\int_{\R^2}J_0(k|x-y|)
S(y)dyN_1(k |z-x|)\cr
&+\int_{\R^2}N_0(k|x-y|)
S(y)dyJ_1(k |z-x|)\Bigg]dk d\theta \cr
=& \frac{1}{8\pi}\int_{\pa B_R(0)}\int _0^{+\infty}k \nu(x)\cdot\Bigg[\int_{\R^2}J_0(k|x-y|)
S(y)dy \nabla_x N_0(k |z-x|)\cr
&+\int_{\R^2}N_0(k|x-y|)
S(y)dy\nabla_xJ_0(k |z-x|)\Bigg]dk ds(x)\cr
=&\frac{1}{8\pi}\int_{0}^{+\infty}k \int_{\pa B_R(0)}\Bigg[\frac{\pa N_0(k |z-x|)}{\nu(x)}\left(\int_{B_R(0)}J_0(k |x-y|)S(y)dy\right)\cr
&-\frac{\pa J_0(k |z-x|)}{\nu(x)}\left(\int_{B_R(0)}N_0(k |x-y|)S(y)dy\right)\Bigg] ds(x) dk \cr
=& S(z),\quad z\in B_R(0)
.
\end{align}
The proof is complete.
\end{proof}


\section{Numerical examples and discussions}\label{Numerical examples}
For numerical implementation, we adopt the following discretization scheme. The sensor positions are uniformly distributed on the circle  $\pa B_5(0)$, forming a sparse array defined as:
\begin{equation}
\Gamma_L:=\left\{5\left(\cos\left(\frac{2\pi l}{L}\right), \sin\left(\frac{2\pi l}{L}\right)\right)\Bigg | l=0,1,\ldots,L-1\right\}
\label{numerical-GammaL}
\end{equation}
where $L$ denotes the number of sensors. The set of wavenumbers used for data acquisition is given by:
\begin{equation}
K:=\left\{k_m=k_{-}+(m-1)dk\big | \, m=1,2,\ldots,N_{\text{max}},\, k_{+}=k_{N_{\text{max}}} \right\}.
\label{numerical-wave-number}
\end{equation}
In all simulations, we set $L=30$ (or $60$), $k_-=0.1$, $k_+=30$ (or $50$) and $dk=0.1$.

Based on the result in Theorem 2.1, we formulate the following  \textbf{Algorithm}  to reconstruct the source function $S(z)$, $z\in \R^2$.

\begin{algorithm}[h!]
\label{main-algorithm}
\caption{Quantitative sampling method for the source function $S(z)$}

1. Collect scattered field patterns $u^s(x,k)$ for all $x\in \Gamma_L$, $k\in [k_-,k_+]$;\\
2. Compute the indicator function $I_{S}(z)$ for all sampling points $z\in D$;\\
3. Plot the indicator function $I_{S}(z)$.
\end{algorithm}

\begin{figure}[h!]
\centering
\begin{tabular}{ccc}
\subfigure[\textbf{Example 1}]{
\includegraphics[width=.25\textwidth]{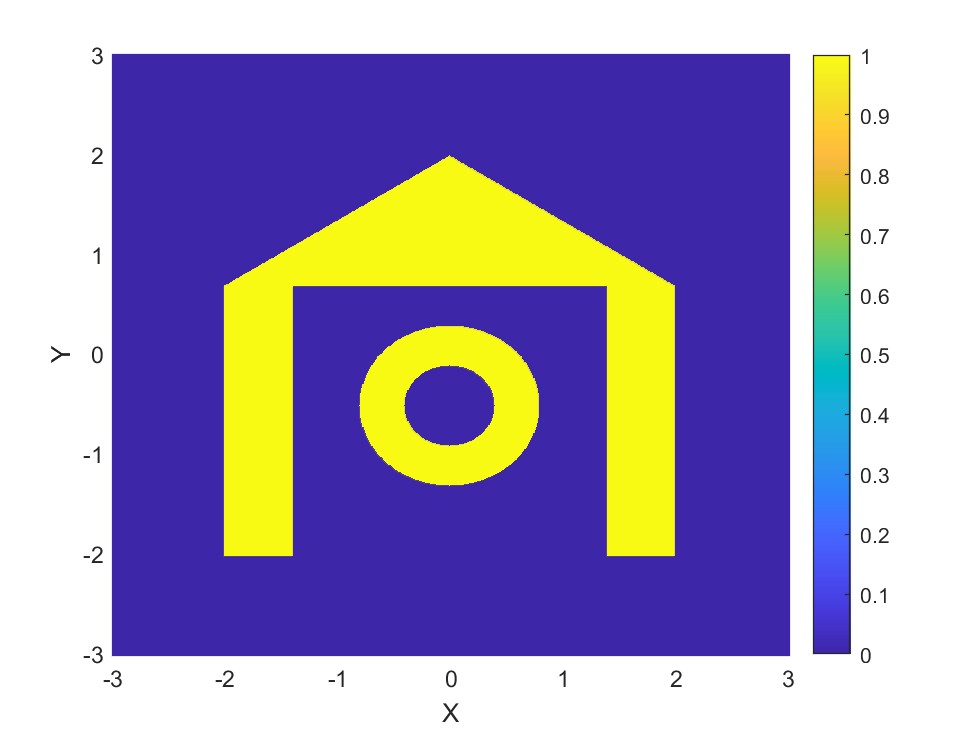}
\label{real-mix}
}\hspace{0em} &  
\subfigure[\textbf{Example 2}]{
\includegraphics[width=.25\textwidth]{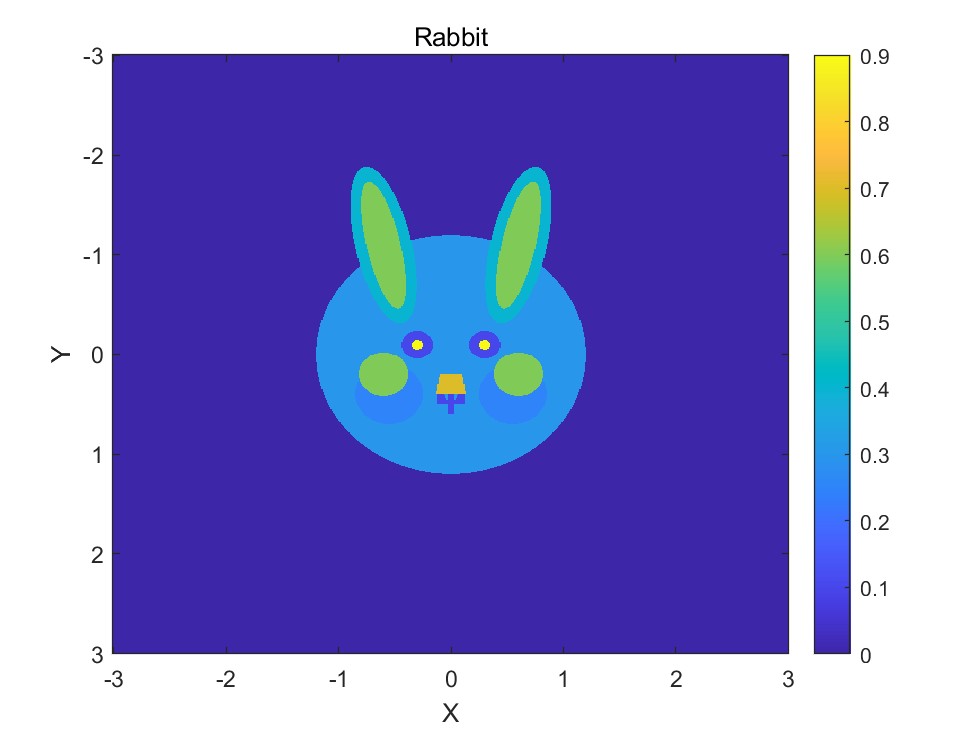}
\label{real-rabbit}
}\hspace{0em} &  
\subfigure[\textbf{Example 3}]{
\includegraphics[width=.25\textwidth]{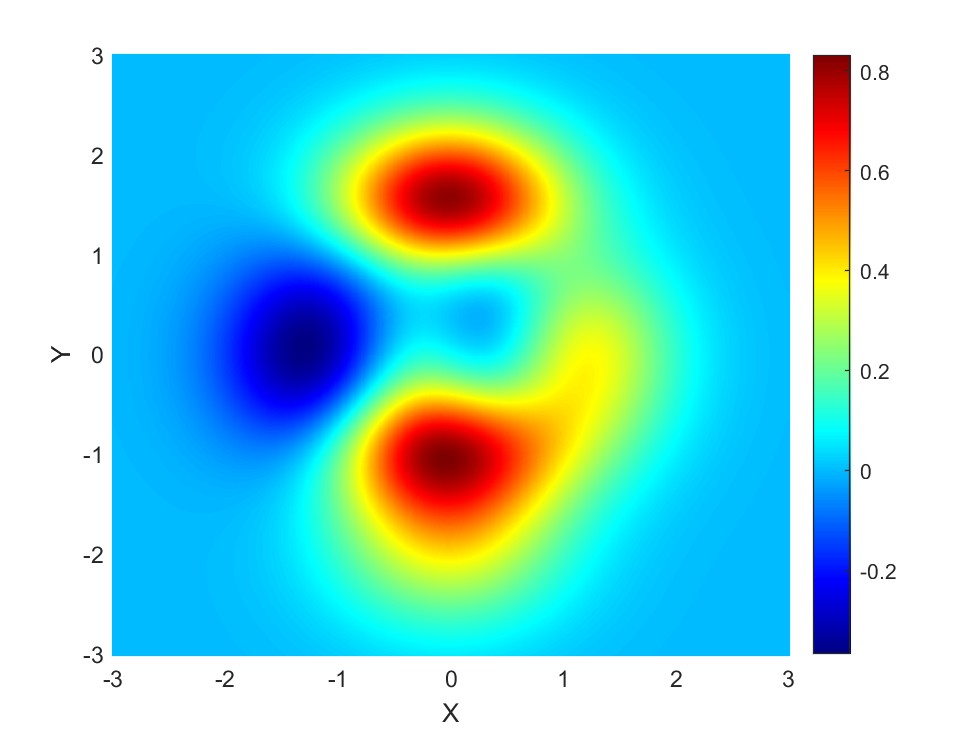}
\label{real-general}
}
\end{tabular}
\caption{The true sources.}
\label{4-real sources}
\end{figure}
The scattered field is calculated using equation \eqref{Scattered field} and then perturbed with noise according to
\begin{equation}
\begin{aligned}
u^{s,\delta}(x_l,k_m)&=u^s(x_l,k_m)(1+\delta \xi_{lm}),\quad x_l\in \Gamma_L,\, k_m\in K,
\label{numerical-u^{s,noise}}
\end{aligned}
\end{equation}
where the relative noise amplitude is taken as $\delta=0.2$ and each $\xi_{lm}$ denotes an independent random variable drawn from a uniform distribution in the interval $[-1,1]$.
Additionally, it is assumed a priori that the source lies inside the search region $D=[-3,3]\times [-3,3]$, which is discretized uniformly with a $601\times 601$ mesh of sampling points $z$.

To verify the effectiveness and robustness of our indicator function $I_S$, we designed the following three examples
\begin{itemize}
    \item \textbf{Example 1}: $S=\chi_{\Omega}$ with $\Om$ being a hybrid source of polygon and annulus in Figure \ref{real-mix};
    \item \textbf{Example 2}: A piecewise constant source function, supported in a rabbit-shaped domain depicted in Figure \ref{real-rabbit};
    \item \textbf{Example 3}: A smooth source function $S(y)$ defined by 
    \begin{equation*}
    \begin{aligned}
     S(y)=&0.3(1 - y_2)^2 e^{-[(y_1)^2 + (y_2 + 1)^2]}- 0.03 e^{-[(y_1 + 1)^2 + (y_2)^2]}\\
     &- [0.3y_1 - (y_1)^3 - (y_2)^5]e^{-[(y_1)^2+ (y_2)^2] },\quad y=(y_1,y_2)\in \R^2,
\end{aligned}
\end{equation*}
as shown in Figure \ref{real-general}.
\end{itemize}


\begin{figure}[h!]
\centering
\begin{tabular}{ccc}
\subfigure[$L=30,k_+=30$]{
\includegraphics[width=.25\textwidth]{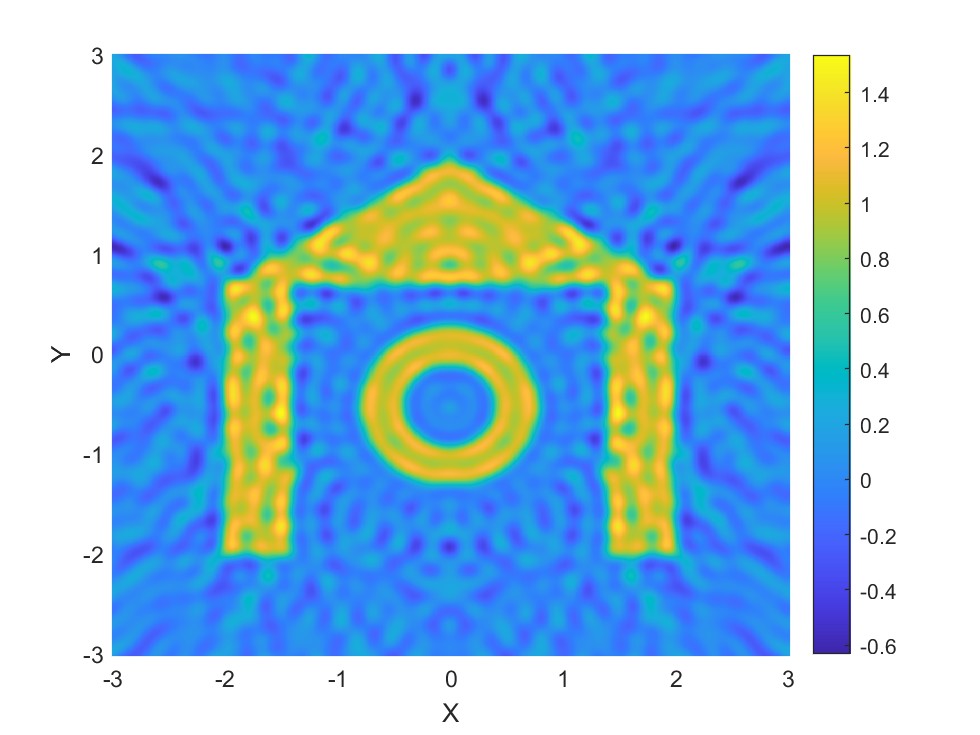}
}\hspace{0em} &
\subfigure[$L=60,k_+=50$]{
\includegraphics[width=.25\textwidth]{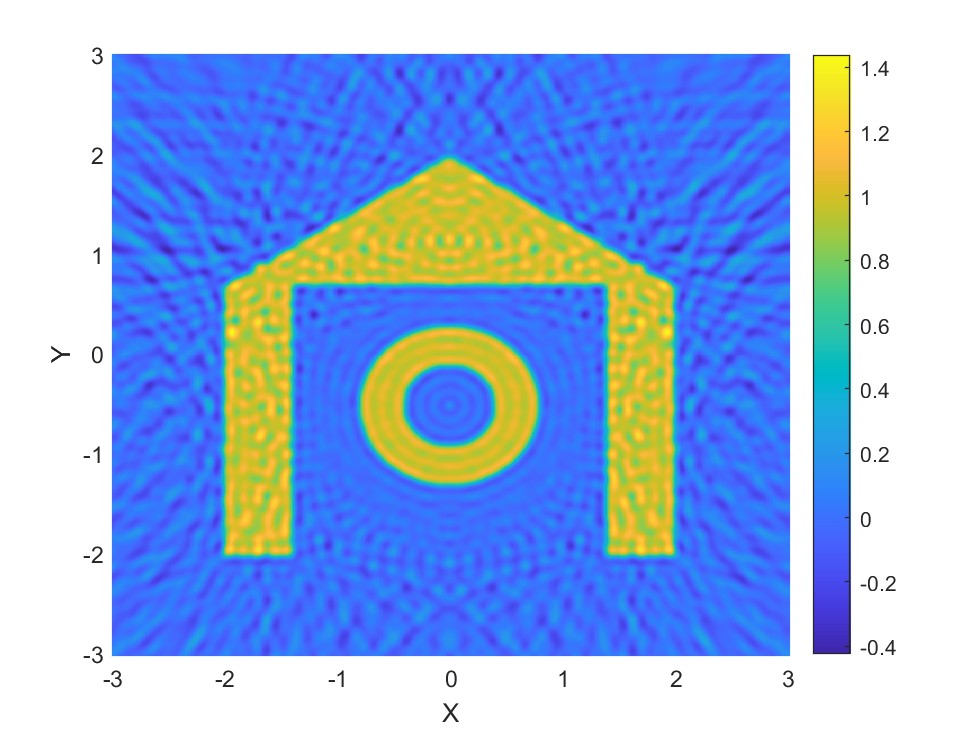}}
\hspace{0em} &
\subfigure[The error $\left| I_S-S\right|$ ($L=60,k_+=50$)]{
\includegraphics[width=.25\textwidth]{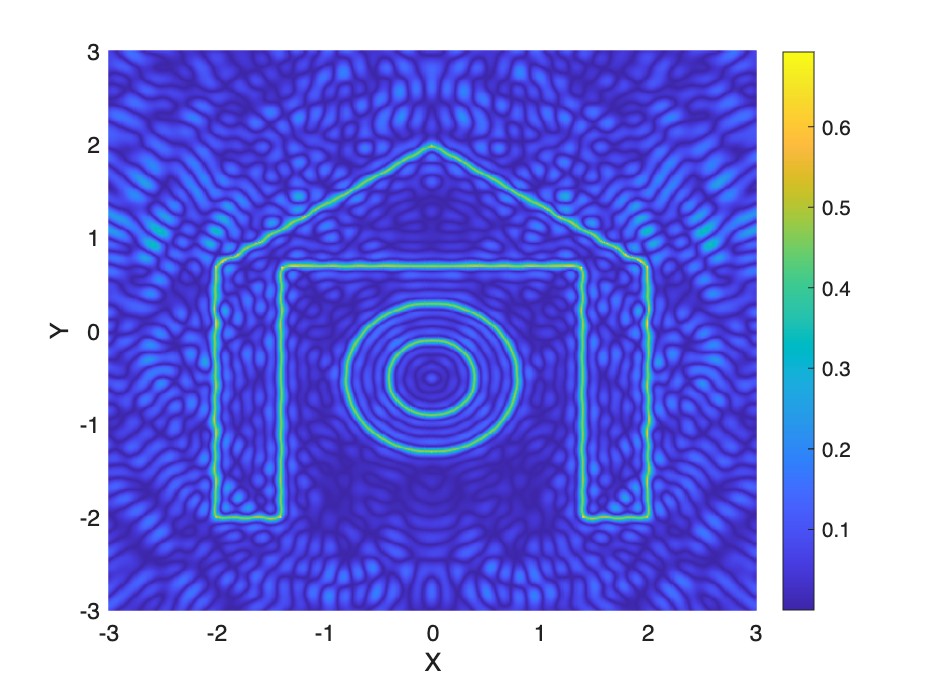}}
\end{tabular}
\caption{Reconstructions of \textbf{Example 1}.}
\label{Mix-pa-source}
\end{figure}
Figure \ref{Mix-pa-source} presents the results of reconstructing the hybrid non-smooth source from Example 1. The plot of $I_S$ successfully captures the geometrical support, which consists of a polygon and an annulus. While the error concentrates near the boundaries, highlighting sensitivity to structural complexity, the interior values of the source function are well retrieved even with $20\%$ noise.

\begin{figure}[h!]
\centering
\begin{tabular}{ccc}
\subfigure[$L=30,k_+=30$]{
\includegraphics[width=.25\textwidth]{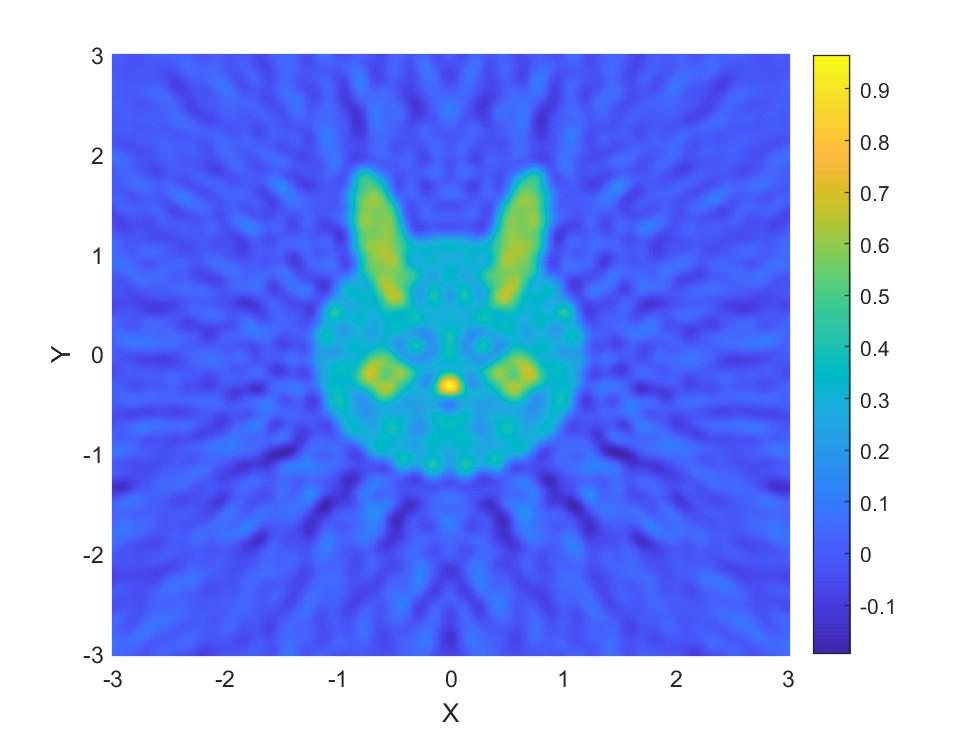}
}\hspace{0em} &
\subfigure[$L=60,k_+=50$]{
\includegraphics[width=.25\textwidth]{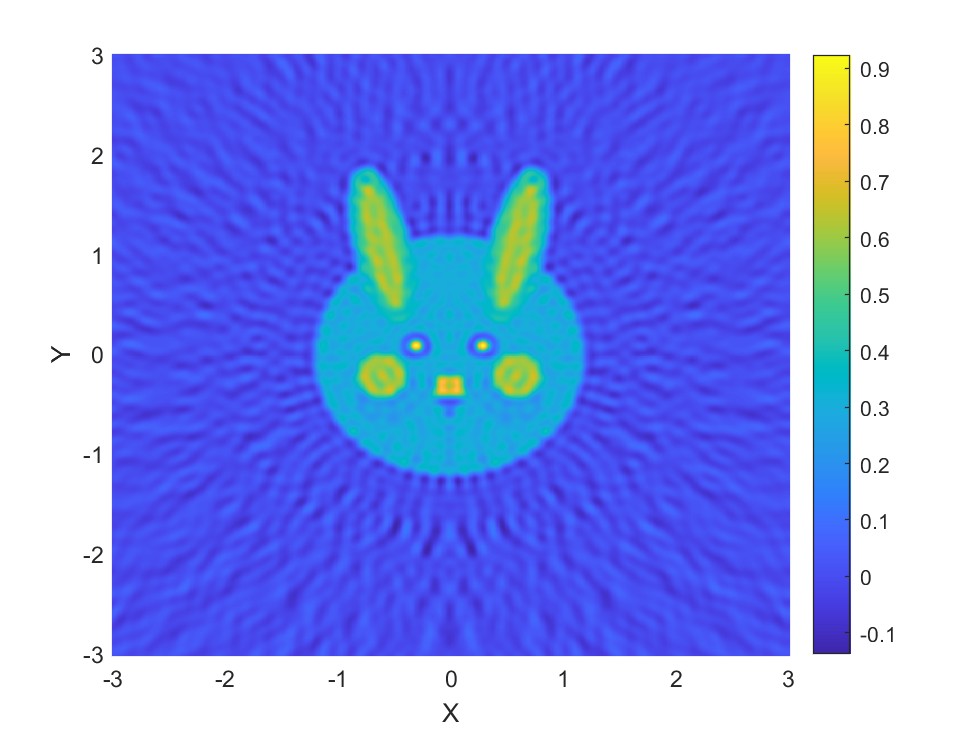}}
\hspace{0em} &
\subfigure[The error $\left| I_S-S\right|$ ($L=60,k_+=50$)]{
\includegraphics[width=.25\textwidth]{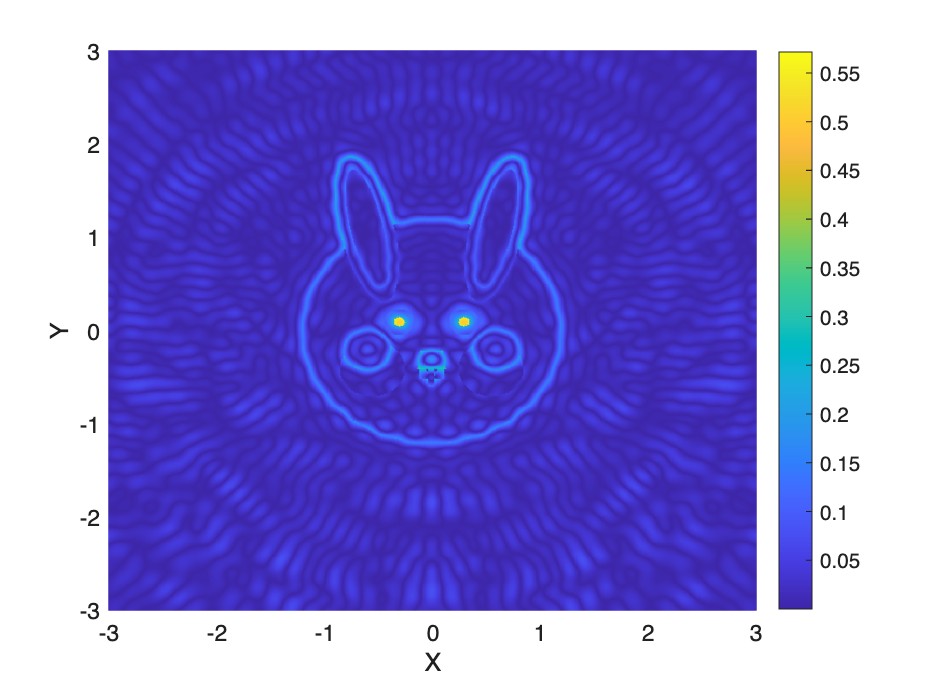}}
\end{tabular}
\caption{Reconstructions of \textbf{Example 2}.}
\label{rabbit-pa-source}
\end{figure}

Figure \ref{rabbit-pa-source} illustrates an example of a piecewise-constant source with sharp discontinuities. Here, we observe that $I_S$ again successfully reconstructs the support domain, including intricate boundaries between constant regions. The method demonstrates robustness against edge artifacts, maintaining transition sharpness even for small or irregularly shaped segments, highlighting its adaptability to complex source profiles.

\begin{figure}[h!]
\centering
\begin{tabular}{ccc}
\subfigure[$L=30,k_+=30$]{
\includegraphics[width=.25\textwidth]{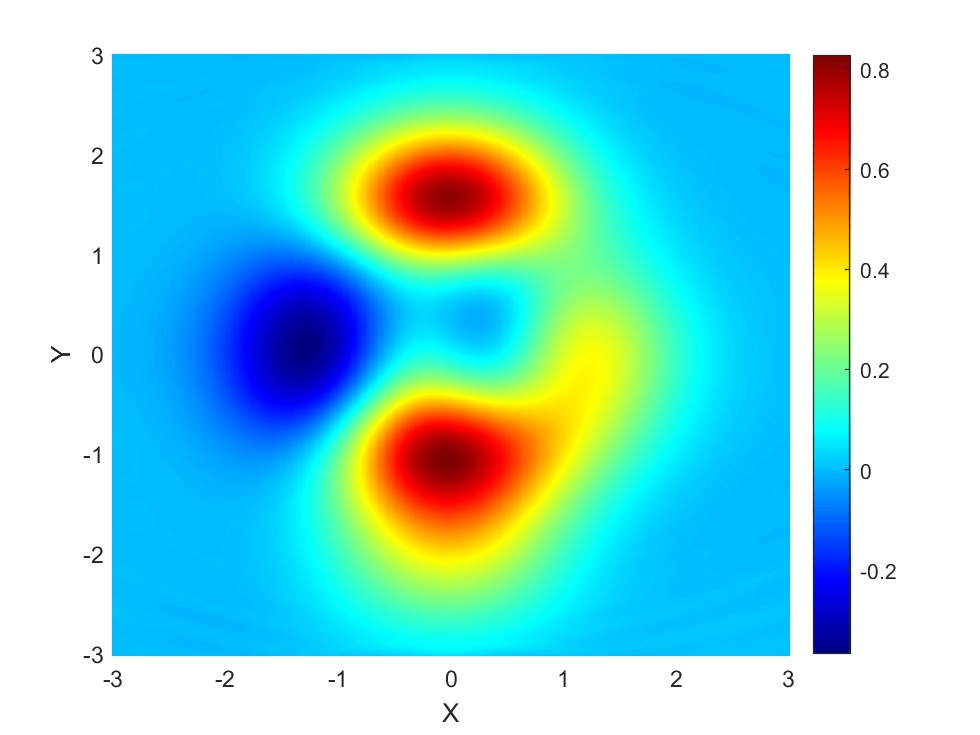}
\label{L30-k-30}
}\hspace{0em} &
\subfigure[$L=60,k_+=30$]{
\includegraphics[width=.25\textwidth]{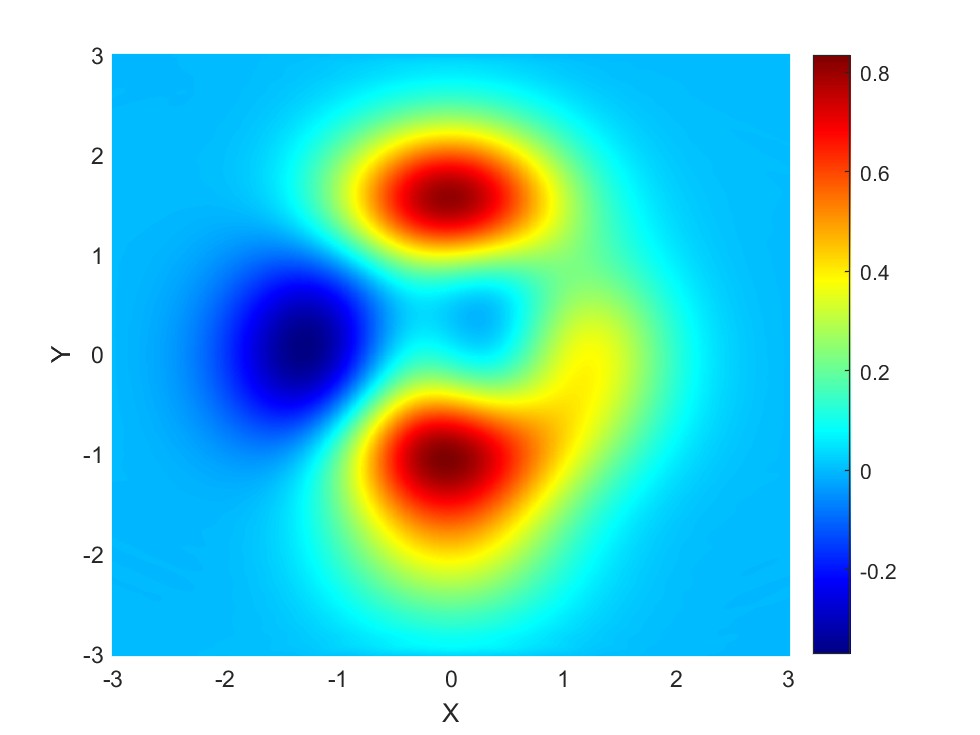}
\label{L60-k-30}
}\hspace{0em} &
\subfigure[The error $\left|I_S-S\right|$ ($L=60,k_+=30$)]{
\includegraphics[width=.25\textwidth]{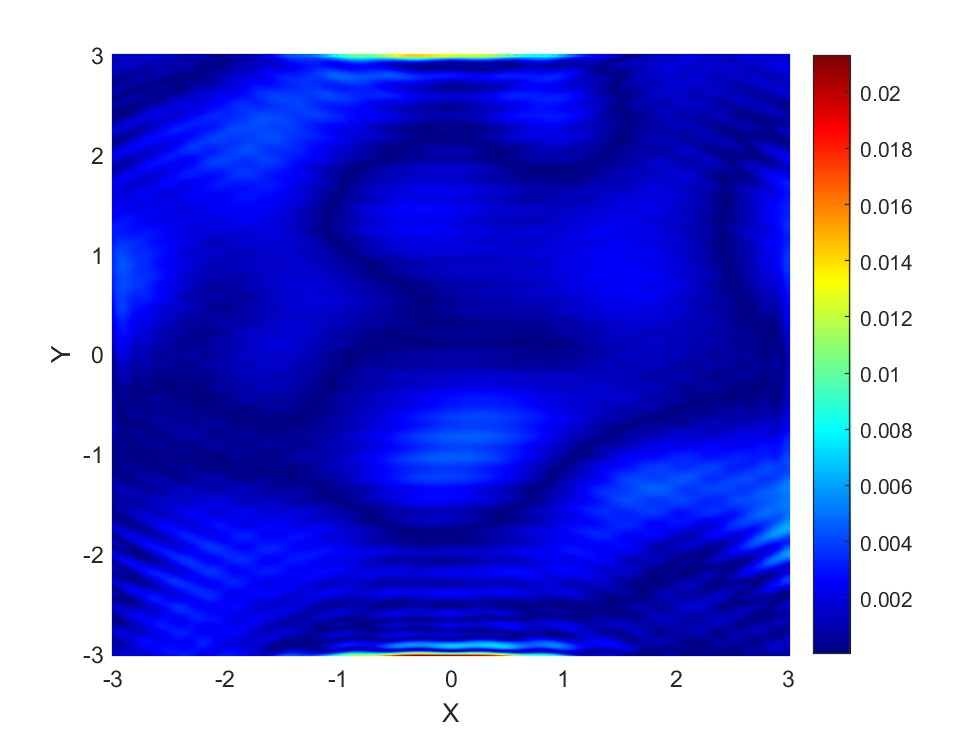}
\label{error-abs}
}
\end{tabular}
\caption{Reconstructions of  \textbf{Example 3}.}
\label{general-source-error}
\end{figure}

Figure \ref{general-source-error} examines the quantitative accuracy of $I_S$ for a smooth, non‑constant source shown in Figure 1(c). Figures \ref{L30-k-30} and \ref{L60-k-30} show the reconstructed source function, which closely matches the amplitude variations, peaks, and gradients of the true source. Figure \ref{error-abs} displays the absolute error $|I_S-S|$. The errors are very small (generally below 0.01), which confirms that $I_S$ recovers function values‌ with high quantitative precision.  This example validates the theoretical equality 
$I_S(z)=S(z)$ established in Theorem 2.1, demonstrating that the Radon-transform‑based formula not only locates the source support but also provides precise amplitude reconstruction.


\section*{Acknowledgment}
The research of X. Liu is supported by the National Key R\&D Program of China under grant 2024YFA1012303  and the NNSF of China under grant 12371430.

\bibliographystyle{SIAM}

\end{document}